\documentclass{llncs}

\newcommand\forFull[1]{#1}
\newcommand\forAbs[1]{}

\newcommand*\samethanks[1][\value{footnote}]{\footnotemark[#1]}

\title{Improved Approximation Algorithms \\for (Budgeted) Node-weighted Steiner Problems}
\author{MohammadHossein Bateni\inst{1}
\and
MohammadTaghi Hajiaghayi\inst{2} \thanks{Supported in part by NSF CAREER award 1053605, NSF grant CCF-1161626,
      ONR YIP award N000141110662, DARPA/AFOSR grant FA9550-12-1-0423,
      and a University of Maryland Research and Scholarship Award (RASA).}
\and
Vahid Liaghat\inst{2} \samethanks
}
\institute{
  Google Research, 76 Ninth Avenue, New York, NY 10011 
\and
  Computer Science Department,
  Univ of Maryland,
  A.V.W. Bldg.,
  College Park, MD 20742
}

\let\proof\undefined

\usepackage{amsfonts,url}
\usepackage{latexsym,graphicx,dsfont}
\usepackage{amsmath,amssymb,enumerate,amsthm}
\usepackage{xspace, times}
\usepackage{bm}
\usepackage{ifpdf}
\usepackage{shadow,shadethm,color}
\usepackage[noend]{algorithmic}
\usepackage{algorithm}
\usepackage{subfigure}
\usepackage{verbatim}

\usepackage
  [pdfpagelabels,plainpages=false,breaklinks,bookmarks,bookmarksnumbered,bookmarksopen,bookmarksopenlevel=2]
  {hyperref}
{\makeatletter \hypersetup{pdftitle={\@title}}}

\usepackage[usenames,dvipsnames]{xcolor}
\usepackage{framed}
\usepackage{tikz}
\usetikzlibrary{matrix,decorations.pathreplacing}

\newif\ifdraft
\draftfalse
\drafttrue 

\ifdraft
\newcounter{Notecount}[section] \newcommand{\theNote}{\thesection.\arabic{Notecount}}
\newcommand{\mhnote}[1]{\refstepcounter{Notecount}\textcolor{Plum}{%
    \mathversion{bold}\marginpar{\hfill\tiny\sffamily\bfseries
      \textcolor{Plum}{MH~\theNote}}$\ll$\bfseries\sffamily#1
    --Mohammad$\gg$}\mathversion{normal}}
\newcommand{\vlnote}[1]{\refstepcounter{Notecount}\textcolor{Red}{%
    \mathversion{bold}\marginpar{\hfill\tiny\sffamily\bfseries
      \textcolor{Red}{VL~\theNote}}$\ll$\bfseries\sffamily#1
    --Vahid$\gg$}\mathversion{normal}}
\reversemarginpar
\else
\newcommand{\vlnote}[1]{}
\newcommand{\mhnote}[1]{}
\fi

\let\realbfseries=\bfseries
\def\bfseries{\realbfseries\boldmath}

\newtheorem{fact}{Fact}

\spnewtheorem*{maintheorem}{Theorem}{\bfseries}{\itshape}

\newcommand{\Mb}[1]{\mathbf{#1}}
\newcommand{\Mc}[1]{\mathcal{#1}}

\newcommand{\y}[0]{\Mb{y}}

\newcommand{\x}[0]{\Mb{x}}
\newcommand{\z}[0]{\Mb{z}}
\newcommand{\f}[0]{\Mb{f}}

\newcommand{\OPT}[0]{\textrm{OPT}}

\newcommand{\Tau}[0]{\Mc{T}}
\newcommand{\eps}[0]{\epsilon}

\newcommand{\core}[0]{\mathbf{core}}
\newcommand{\pcsf}[0]{\mathbf{PCSF}}

\newcommand{\ignore}[1]{}

\def\caps#1{\textsc{#1}}
\def\MST{\caps{MST}}
\def\ST{\caps{Steiner Tree}}

\newcommand{\todo}[1]{}

\usepackage{hyperref}
\hypersetup{
    colorlinks,
    citecolor=blue,
    filecolor=black,
    linkcolor=magenta,
    urlcolor=black
}

\begin{document}
\maketitle

\begin{abstract}
Moss and Rabani~\cite{MossRabani} study constrained node-weighted Steiner tree problems
with two independent weight values associated with each node, namely, cost and prize (or penalty).
They give an $O(\log n)$-approximation algorithm for the prize-collecting node-weighted Steiner
tree problem (PCST)---where the goal is to minimize the cost of a tree plus the penalty of vertices
not covered by the tree. They use the algorithm for PCST to obtain a bicriteria
$(2, O(\log n))$-approximation algorithm for the Budgeted node-weighted Steiner tree
problem---where the goal is to maximize the prize of a tree with a given budget for its cost.
Their solution may cost up to twice the budget, but collects a factor
$\Omega(\frac{1}{\log n})$ of the optimal prize. We improve these results 
from at least two aspects.

Our first main result is a primal-dual $O(\log h)$-approximation algorithm for a more general
problem, prize-collecting node-weighted Steiner forest (PCSF), where we have $h$ demands each
requesting the connectivity of a pair of vertices.
Our algorithm can be seen as a greedy algorithm which reduces the number of demands by choosing a structure with minimum cost-to-reduction ratio.
This natural style of argument (also used by Klein and Ravi~\cite{KleinR95} and Guha et al.~\cite{Guha99}) leads to a much simpler algorithm than that of
Moss and Rabani~\cite{MossRabani} for PCST.


Our second main contribution is for the Budgeted node-weighted Steiner tree problem, which is
also an improvement to Moss and Rabani~\cite{MossRabani} and Guha et al.~\cite{Guha99}.
In the unrooted case, we improve upon an $O(\log^2 n)$-approximation of \cite{Guha99}, and present
an $O(\log n)$-approximation algorithm without any budget violation. For the rooted case, where a
specified vertex has to appear in the solution tree, we improve the bicriteria result of
\cite{MossRabani} to a bicriteria approximation ratio of $(1+\eps, O(\log n)/\eps^2)$ for any
positive (possibly subconstant) $\eps$. That is, for any permissible budget violation $1+\eps$,
we present an algorithm achieving a tradeoff in the guarantee for prize. Indeed, we show that this
is almost tight for the natural linear-programming relaxation used by us as well as in \cite{MossRabani}.

\end{abstract}


\section{Introduction}
In the rapidly evolving world of telecommunications and internet, design of fast and efficient networks is of utmost importance.
It is not surprising, therefore, that the field of network design has continued to be an active area of research since its inception several decades ago.
These problems have applications not only in designing computer and telecommunications networks,
but are also essential for other areas such as VLSI design and computational geometry~\cite{SteinerIndustryBook}.
Besides their appeals in these applications, basic network design problems (such as Steiner Tree, TSP, and their variants)
have been the testbed for new ideas and have been instrumental in development of new techniques in the field of approximation algorithms.

In parallel to the study by Moss and Rabani~\cite{MossRabani},
this work focuses on graph-theoretic problems in which two (independent) nonnegative weight functions
are associated with the vertices, namely cost $c(v)$ and prize (or penalty) $\pi(v)$ for each vertex
$v$ of the given graph $G(V, E)$.
The goal is to find a connected subgraph $H$ of $G$ that optimizes a certain objective.
We now summarize the four different problems, already introduced in the literature.
In the \emph{Net Worth} problem (NW), the goal is to maximize the prize of $H$ minus its cost%
\footnote{The prize or cost of a subgraph is defined as the total prize or cost of its vertices, respectively}.
\forFull{We prove in Appendix~\ref{sec:nw-hardness} that this natural problem does not admit any finite approximation algorithm.}
\forAbs{It can be proved that this natural problem does not admit any finite approximation algorithm (see the full version of this work).}
A similar, yet better-known objective is that of minimizing the cost of the subgraph plus the
penalty of nodes outside of it (which is called \emph{Prize-Collecting Steiner Tree} (PCST) in the
literature).
Two other problems arise if one restricts the range of either cost or prize in the desired solution.
In particular, the \emph{Quota} problem
tries to find the minimum-cost tree among those with a total prize
surpassing a given value, whereas the \emph{Budgeted} problem deals with maximizing the prize with a given
maximum budget for the cost.
The rooted variants ask, in addition, that a certain root vertex be included in the solution.
In the $k$-\MST{} problem, the goal is to find a minimum-cost tree with at least $k$ vertices. In the $k$-\ST{} problem, given a set of terminals, the goal is to find a minimum-cost tree spanning at least $k$ terminals.
We show the following reductions missing from the literature.
\begin{theorem}
Let $\alpha$, $0<\alpha<1$, be a constant. The following statements are equivalent (both for edge-weighted and node-weighted variants):
\begin{enumerate}[i]
\item There is an $\alpha$-approximation algorithm for the rooted $k$-\MST{} problem.
\item There is an $\alpha$-approximation algorithm for the unrooted $k$-\MST{} problem.
\item There is an $\alpha$-approximation algorithm for the $k$-\ST{} problem.
\end{enumerate}
\end{theorem}
\begin{proof}
\forAbs{Here we present the equivalence of (ii) and (iii) (see the full version of this work for that of (i) and (ii)).}
\forFull{Here we present the equivalence of (ii) and (iii) (see Appendix~\ref{sec:reductions} for that of (i) and (ii)).}
We note that one way is clear by definition. To prove that (iii) implies (ii), we give a cost-preserving reduction from $k$-\ST{} to $k$-\MST{}. Let $<G=(V,E),T,k>$ be an instance of $k$-\ST{} with the set of terminals $T\subseteq V$. Let $n=|V|$. For every terminal $v_t\in T$, add $n$ vertices at distance zero of $v_t$. Let $k'=kn+k$ and consider the solution to $k'$-\MST{} on the new graph. Any subtree with at most $k-1$ terminals have at most $(k-1)n+n-1=kn-1$ vertices. Therefore an optimal solution covers at least $k$ terminals. Hence the reduction preserves the cost of optimal solution.
\end{proof}
These results improve the approximation ratio for $k$-Steiner tree.
Previously, a $4$-approximation algorithm was proved by \cite{RaviEtAl}
and a $5$-approximation algorithm was due to \cite{CRW} who had also conjectured
the presence of a $2+\eps$-approximation algorithm.
The equivalence of $k$-Steiner tree and $k$-MST
combined with the $2$-approximation result of Garg~\cite{Garg05} leads to a $2$-approximation
algorithm for $k$-Steiner tree.

A more tractable version of the prize-collecting variant is the edge-weighted case in which the costs (but not the
prizes) are associated with edges rather than nodes.
The best known approximation ratio for the edge-weighted Steiner tree problem is $1.39$ due to Byrka et al.~\cite{ByrkaGRS10}.
For the earlier work on edge-weighted variant we refer the reader to the references of \cite{ByrkaGRS10}.
In this paper, unless otherwise specified all our graphs are node-weighted and undirected.

\subsection{Contributions and Techniques}
\paragraph{Approximation algorithm for PCSF.}
%
Klein and Ravi~\cite{KleinR95} were the first to give an $O(\log h)$-approximation algorithm for the SF problem. Later,
Guha et al.~\cite{Guha99} improved the analysis of \cite{KleinR95} by showing that the approximation ratio of the algorithm of \cite{KleinR95} is w.r.t. the fractional optimal solution for the ST problem. The ST problem is a special case of SF where all demands share an endpoint.
Very recently and independently of our work, Chekuri et al.~\cite{ChekuriEV12} give an algorithm with an approximation ratio of $O(\log n)$ w.r.t. to the fractional solution for SF and higher connectivity problems. This immediately provides a reduction from PCSF to the SF problem: one can fractionally solve the LP for PCSF and pay the penalty of every demand for which the fractional solution pays at least half its penalty. Hence, the remaining demands can be (fractionally) satisfied by paying at most twice the optimal solution. Therefore, one can make a new instance of SF with only the remaining demands and get a solution within $O(\log n)$ factor of the optimal solution using the SF algorithm.

We start off by presenting a simple primal-dual $O(\log h)$-approximation algorithm
for the node-weighted prize-collecting Steiner forest (PCSF) problem where $h$ is the number of connectivity demands---see Theorem~\ref{thm:gwalg}.
Compared to the PCST algorithm given by Moss and Rabani~\cite{MossRabani} and Konemann et al.~\cite{Sina13},
our algorithm for PCSF solves a more general problem and it has a simpler analysis.
A reader familiar with the moat-growing framework%
\footnote{Introduced by Agrawal, Klein, and Ravi (AKR)~\cite{AKR} and Goemans and Williamson (GW)~\cite{GW92}.}
may recall that algorithms in this framework (e.g., that of Moss and Rabani~\cite{MossRabani} or K\"onemann et al.~\cite{Sina13}) consist of a \emph{growth phase} and a \emph{pruning phase}.
A moat is a set of dual variables corresponding to a laminar set of vertices containing \emph{terminals}---vertices with a positive penalty. The algorithm grows the moats by increasing the dual variables and adding other vertices gradually to guarantee feasibility.
In the edge-weighted Steiner tree problem, when two moats collide on an edge, the algorithm buys the path connecting the moats and merges the moats.
Roughly speaking, the algorithm stops growing a moat when either it reaches the root, or its total growth reaches the total prize of terminals inside it.
This process is not quite enough to obtain a good approximation ratio. At the end of the algorithm we may have paid too much for connecting unnecessary terminals. Thus as a final step one needs to prune the solution in a certain way to obtain the tight approximation ratio of $2-\frac{1}{n}$.

In the node-weighted problem, one obstacle is that (polynomially) many moats may collide on a vertex. Handling the proper growth of the moats and the process to merge them proves to be very sophisticated. This may have been the reason that for more than a decade no one noticed
the flaw in the algorithm of Moss and Rabani~\cite{MossRabani}%
\footnote{In private correspondence the authors of the
original work have admitted that their algorithm is flawed and that it cannot be fixed easily.
}.
Indeed the recently proposed algorithm by K\"onemann et al.~\cite{Sina13} is even more sophisticated.
In our algorithm, not only do we completely discard the pruning phase,
but we also never merge the moats (thus intuitively, a moat forms a disk centered at a terminal).
In fact, our algorithm can be thought of as a simple greedy algorithm.
Our algorithm runs in iterations, and in each iteration several disks are grown simultaneously on different
endpoints of the demands.
The growth stops at the largest possible radius where there are no
``overlaps'' and no disk has run out of ``penalty.''
If the disks corresponding to several endpoints hit each other,
a set of paths connecting them is added to the solution and
all but one representative endpoint
are removed for the next iteration.
However, if a disk is running out of penalty, the terminal at its center is removed for the next iteration.
The cost incurred at each iteration is a fraction of $\OPT$,
 proportional to the fraction of endpoints removed,
hence the logarithmic term in the guarantee.

Although our primal-dual approach is different from the approach known for SF~\cite{KleinR95,Guha99}, we indeed use the same style of argument to analyze our algorithm. The crux of these algorithms is to reduce the number of components of the solution by using a structure with minimum cost-to-reduction ratio. Besides the simplicity of this trend, it is important that by avoiding the pruning phase, these algorithms may lead to progress in related settings such as streaming and online settings.
The moat-growing approach of Konemann et al.~\cite{Sina13}, however,  allows a stronger lagrangian-preserving guarantee
\footnote{Let $T$ denote the sets of vertices purchased by the algorithm of \cite{Sina13}. It is guaranteed that $c(T)+\log(n)\pi(V\backslash T)\leq \log(n) \OPT$.}
for PCST. This property is shown to be quite important for solving various problems such as $k$-MST and $k$-Steiner tree (see e.g. \cite{CRW,JainV01}).

\paragraph{Approximation algorithms for the Budgeted problem.}
Using their algorithm for PCST, Moss and Rabani developed a bicriteria\footnote{An $(\alpha,\beta)$-bicriteria
approximation algorithm for the Budgeted problem finds a tree with total prize at least
$\frac{1}{\beta}$ fraction of that of optimal solution and total cost at most $\alpha$ factor of
the budget.} approximation algorithm for the Budgeted problem,
one that achieves an approximation factor $O(\log n)$ on prize while violating the budget constraint by no more than factor two~\cite{MossRabani}.
We present in Theorem~\ref{thm:rootedeps} a modified pruning procedure that improves the bicriteria bound to $(1+\eps, O(\log n)/\eps^2)$;
in other words, if the algorithm is allowed to violate the budget constraint by only a factor $1+\eps$ (for any positive $\eps$),
the approximation guarantee on the prize will be $O(\log n)/\eps^2$.
In fact, we also show using the natural linear-programming relaxation (used in \cite{MossRabani} as well),
that it is not possible to improve these bounds significantly---
\forFull{see Appendix~\ref{sec:natural-lp}.}
\forAbs{see the full version.}
In particular, there are instances for which the fractional solution is $\OPT/\eps$,
however, no solution of cost at most $1+\eps$ times the budget has prize more than $O(\OPT)$.
Our integrality-gap construction fails if the instance is not rooted.
Indeed, in that case, we show how to obtain an $O(\log n)$-approximation algorithm with
no budget violations---see Theorem~\ref{thm:logunrooted}.
This improves the $O(\log^2n)$-approximation algorithm of Guha et al.~\cite{Guha99}.\footnote{%
The $O(\log^2n)$-approximation algorithm can be derived from the results in \cite{Guha99} with
some efforts, not as explicitly as cited by Moss and Rabani~\cite{MossRabani}.}
To get over the integrality gap of the LP formulation,
we prove several structural properties for near-optimal solutions.
By restricting the solution to one with these properties,
we use a bicriteria approximation algorithm 
as a black box to find a near-optimal solution.
Finally we use a generalization of the trimming method of \cite{Guha99}
to avoid violating the budget.

\forFull{
\subsection{Organization}
Next in Section~\ref{sec:mossrabani} we briefly discuss the method of Moss and Rabani
for deriving an algorithm for the budget problem from that for PCST.
We then explain and analyze our algorithm for PCSF in Section~\ref{sec:pcsf}.
Section~\ref{sec:budgeted} discusses our trimming procedure
and how it leads to improved results for Budgeted problems.
Finally, the appendices contain minor results for hardness of NW and reductions
between special cases of the Quota problem, as well as omitted proofs.
} 

\section{The Prize-Collecting Steiner Forest Problem}\label{sec:mossrabani}
The starting point of the algorithm of Moss and Rabani~\cite{MossRabani} is a standard LP relaxation for the rooted version. For the Quota and Budgeted problems they show that any (fractional) feasible solution can be approximated by a convex combination of sets of nodes connected (integrally) to the root. Given the support of such a convex combination, it follows from an averaging argument that a proper set can be found. Thus the problem comes down to finding the support of the convex combination. They show that given a black-box algorithm which solves the PCST problem with the approximation factor $O(\log n)$, one can obtain the support in polynomial time.
%



The main result of this section is a very simple, and maybe more elegant algorithm for
the classical problem of PCSF (and thus PCST). As mentioned before, using moats and having
a pruning phase lead to the main difficulty in the analysis of previous algorithms.
These seem to be a necessary evil for achieving a tight constant approximation factor for
the edge-weighted variant.
Surprisingly, we show neither is needed in the node-weighted variant.
Instead of moats, we use dual \emph{disks} which are centered on a \emph{single} terminal and we do not need a pruning phase.


\subsection{Preliminaries}
Consider a graph $G=(V,E)$ with a node-weight function $c:V\to \mathds{R}_{\geq 0}$. For a subset $S\subseteq V$, let $c(S):=\sum_{v\in S} c(v)$. In the \emph{Steiner Forest} problem, given a set of demands $\Mc{L}=\langle (s_1,t_1),\ldots,(s_h,t_h) \rangle$, the goal is to find a set of vertices $X$ such that for every demand $i\in[h]$, $s_i$ and $t_i$ are connected in $G[X]$. The vertices $s_i$ and $t_i$ are denoted as the \emph{endpoints} of the demand $i$. In PCSF a penalty (prize) $\pi_i\in \mathds{R}_{\geq 0}$ is associated with every demand $i\in[h]$. If the endpoints of a demand are not connected in the solution, we need to pay the penalty of the demand.
The objective cost of a solution $X\subseteq V$ is
\vspace{-0.2cm}
\begin{align*}
\pcsf(X)=c(X)+ \hspace{-3mm}\sum_{\hspace{-3mm}i\in[h] : i\textit{ is not satisfied}} \hspace{-6mm} \pi_i.
\end{align*}
\vspace{-0.2cm}

A \emph{terminal} is a vertex which is an endpoint of a demand. Let $\Tau$ denote the set of terminals. We may assume that the cost of a terminal is zero. We also assume the endpoints of all demands are different%
\footnote{Both assumptions are without loss of generality. For every demand $(s_i,t_i)$, attach a new vertex $s^i$ of cost zero to $s_i$ and similarly attach a new vertex $t^i$ of cost zero to $t_i$. Now interpret $i$ as the demand between $s^i$ and $t^i$. The optimal cost does not change.}
(thus $|\Tau|=2h$).
%
For a pair of vertices $u$ and $v$ and a cost function $c$, let $d^c(u,v)$ denote the length of the shortest path with respect to $c$ connecting $u$ and $v$, including the cost of endpoints.

For a set of vertices $S$ let $\delta(S)$ denote the set of vertices that are not in $S$ but have neighbors in $S$.
A set $S$ \emph{separates} a demand $i$ if exactly one of $s_i$ and $t_i$ is in $S$. Let $\Mc{S}_i$ denote the collection of sets separating the demand $i$ and let $\Mc{S}=\bigcup_i \Mc{S}_i$. For a set $S$, define the penalty of $S$ as half of the total penalty of demands separated by $S$, i.e., $\pi_{\Mc{L}}(S)=\frac{1}{2}\sum_{i: S\in \Mc{S}_i} \pi_i$. We may drop the index $\Mc{L}$ when there is no ambiguity.
The PCSF problem can be formulated as the following standard integer program (IP):
\vspace{-0.2cm}
\begin{align*}
    & \text{Minimize}          & & \sum_{v\in V\backslash \Tau} c(v) \x(v) + \sum_{S\in \Mc{S}} \pi(S) \z(S)\\
    & \forall i\in [h], S\in \Mc{S}_i       & & \sum_{v\in \delta(S)} \x(v) + \sum_{R | S\subseteq R \in \Mc{S}_i} \z(R)  \geq 1 \\
    &                               & & \x(v),\z(S)\in \{0,1\}
\end{align*}
Given a solution $X\subseteq V$ to the PCSF problem one can easily make a feasible solution $\x$ to the IP with the same objective value as $\pcsf(X)$: since the cost of a terminal is zero, we assume $\Tau\subseteq X$. For every vertex $v\in X$ set $\x(v)=1$ and for every connected component $CC$ of $G[X]$ set $\z(V\backslash CC)=1$. It is also easy to verify since the cost of a terminal is zero, any (integral) feasible solution $\x$ corresponds to a solution $X\subseteq V$ for the PCSF problem with (at most) the same cost.
One may relax the IP by allowing assignment of fractional values to the variables. Let $\OPT$ denote the objective value of the optimal solution for the relaxed linear program (LP).
The following is the dual program~\ref{lp:dual} corresponding to the relaxed LP.
\vspace{-0.2cm}
\begin{align}
    & \text{Maximize}                  & & \sum_{S\in \Mc{S}} \y(S)  \tag{$\Mc{D}$} \label{lp:dual}\\
    & \forall v\in V     & & \sum_{S\in\Mc{S}: v\in \delta(S)} \y(S)\leq c(v) \notag \\
    & \forall S\in \Mc{S}             & & \sum_{S'\subseteq S} \sum_{i: S,S'\in \Mc{S}_i} \y_i(S') \leq \pi(S) \notag \\
    &                               & & \y_i(S)\geq 0, \y(S)=\sum_{i: S\in \Mc{S}_i} \y_i(S) \notag
\end{align}
In the case of Steiner tree, the dual variables are defined w.r.t. a set $S$. However, in Steiner forest, the dual variables are in the form $\y_i(S)$, i.e., they are defined based on a demand as well. This has been one source of the complexity of previous primal-dual algorithms for Steiner forest problems. Interestingly, in our approach, we only need to work with a \emph{simplified dual} constructed as follows.
\paragraph{Cores and Simplified Duals} Let $c$ and $\Mc{L}$ denote a node-weight function and a set of demands, respectively. Let $Z_c$ denote the set of vertices with zero cost. We note that the terminals are in $Z_c$. A set $C\subseteq V$ is a \emph{core} if $C$ is a connected component of $G[Z_c]$ \emph{and} contains a terminal (i.e., an endpoint of a demand in $\Mc{L}$). Let $\overline{\Mc{S}(c,\Mc{L})}$ be the collection of sets separating one core from the other cores, i.e., a set $S$ is in $\overline{\Mc{S}(c,\Mc{L})}$ if $S$ contains a core but has no intersection with other cores. For a set $S\in \overline{\Mc{S}(c,\Mc{L})}$, let $\core(S)$ denote the core inside $S$. Note that $\pi_{\Mc{L}}(S)=\pi_{\Mc{L}}(\core(S))$.
A simplified dual w.r.t. $c$ and $\Mc{L}$ is the following program \ref{lp:simdual}.
\vspace{-.2cm}
\begin{align}
    & \text{Maximize}                  & & \sum_{S\in \Mc{S}} \y(S)  \tag{$\overline{\Mc{D}(c,\Mc{L})}$} \label{lp:simdual}\\
    & \forall v\in V     & & \sum_{S\in\overline{\Mc{S}(c,\Mc{L})}: v\in \delta(S)} \y(S)\leq c(v) \tag{C1}\label{eq:dualconst} \\
    & \forall S\in \overline{\Mc{S}(c,\Mc{L})}             & & \sum_{S': \core(S)\subseteq S' \subseteq S} \y(S') \leq \pi_{\Mc{L}}(S) \tag{C2}\label{eq:prizeconst} \\
    &                               & & \y(S)\geq 0 \notag
\end{align}

Observe that $\overline{\Mc{S}(c,\Mc{L})}\subseteq \Mc{S}$. Indeed \ref{lp:simdual} is the same as \ref{lp:dual} with only (much) fewer variables. Thus the program \ref{lp:simdual} is only more restricted than \ref{lp:dual}. In the rest of the paper, unless specified otherwise, by a dual we mean a simplified dual. When clear from the context, we may omit the indices $c$ and $\Mc{L}$.

\paragraph{Disks} Consider a dual vector $\y$ initialized to zero. A \emph{disk of radius $R$ centered at a terminal $t$} is the dual vector obtained from the following process: Initialize the set $S$ to the core containing $t$.
Increase $\y(S)$ until for a vertex $u$ the dual constraint~\ref{eq:dualconst} becomes tight. Add $u$ to $S$ and repeat with the new $S$. Stop the process when the total growth (i.e., sum of the dual variables) reaches $R$.
A disk is \emph{valid} if $\y$ is feasible. In what follows, by a disk we mean a valid disk unless specified otherwise.

A vertex $v$ is \emph{inside} the disk if $d^c(t,v)$ is strictly less than $R$. The \emph{continent} of a disk is the set of vertices inside the disk.
Further, we say a vertex $v$ is on the \emph{boundary} of a disk if it is not inside the disk but has a neighbor $u$ such that $d^c(t,u)\leq R$. Note that $u$ is not necessary inside the disk.
\forFull{See Figure~\ref{fig:disk} for a graphical representation of a disk.}
The following facts about a disk of radius $R$ centered at a terminal $t$ can be derived from the definition:

\forFull{
\begin{figure}[!h]
\begin{center}
    \includegraphics[width=0.8\textwidth]{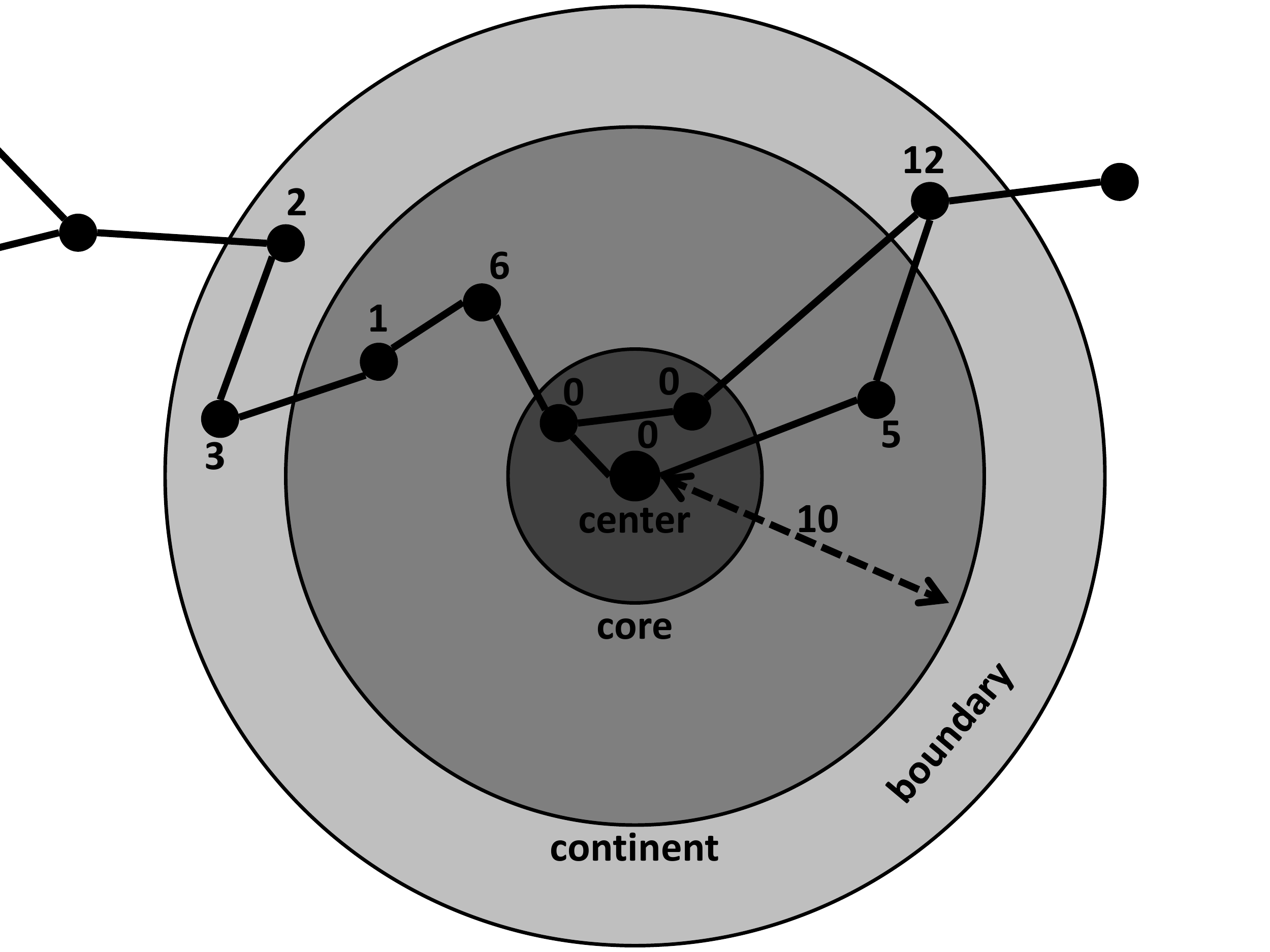}
  \end{center}
  \caption{A graphical representation of a disk of radius $10$. The vertex at the center of the disk is an endpoint of a demand. The numbers show the cost of vertices. The inner-most circle contains the core, while the outer-most circle contains the continent \emph{and} the boundary.}
\label{fig:disk}
\end{figure} 
}

\begin{fact}\label{fact:diskvalue}
The (dual) objective value of the disk is exactly $R$.
\end{fact}

\begin{fact}\label{fact:inside}
For every vertex inside the disk, the dual constraint~\ref{eq:dualconst} is tight.
\end{fact}


\begin{fact}\label{fact:outside}
If a set $S$ does not include the center, then $\y(S)=0$. Further, if $S$ is not a subset of the continent, then $\y(S)=0$.
\end{fact}

Let $\y_1,\ldots,\y_k$ denote a set of disks. The \emph{union} of the disks is simply a dual vector $\y$ such that $\y(S)=\sum_i \y_i(S)$ for every set $S\subseteq \Mc{S}$.
A set of disks are \emph{non-overlapping} if their union is a feasible dual solution (i.e., both set of constraints~\ref{eq:dualconst} and \ref{eq:prizeconst} hold). If a vertex $v$ is inside a disk, the corresponding dual constraint is tight. Thus for any set $S$ such that $v\in\delta(S)$, the dual variable $\y(S)$ cannot be increased. On the other hand since the distance between $v$ and the center is \emph{strictly} less than the radius, there exists a set containing $v$ with positive dual value.
This observation leads to the following.

\begin{proposition}\label{prop:inorout}
Let $\y$ be the union of a set of non-overlapping disks $\y_1,\ldots,\y_k$. A vertex inside a disk cannot be on the boundary of another disk.
\end{proposition}

Proposition~\ref{prop:inorout} implies that in the union of a set of non-overlapping disks, the continents are pairwise far from each other.
This intuition leads to the following
\forFull{(proof in Appendix~\ref{sec:omitted}).}
\forAbs{.\footnote{Due to the lack of space, we have omitted some of the proofs throughout the paper. We defer the reader to the full version of this paper for the omitted proofs.}}


\begin{lemma}\label{lem:maxR}
    Suppose $T'$ is a subset of terminals such that the distance between every pair of them is non-zero.
    Let $R$ denote the maximum radius such that the $|T'|$ disks of radius $R$ centered at terminals in $T'$ are non-overlapping. Consider the union of such disks. Either (i) the constraint~\ref{eq:prizeconst} is tight for a continent; or (ii) the constraint~\ref{eq:dualconst} is tight for a vertex on the boundary of multiple disks.
\end{lemma}

The final tool we need for the analysis of the algorithm states a precise relation between the dual variables and the distance of a vertex on the boundary.
\forFull{The proof is based on the analysis of the growth of a disk which is presented in Appendix~\ref{sec:omitted}.}

\begin{lemma}\label{lem:boundarydistance}
Let $v$ be a vertex on the boundary of a disk $\y$ of radius $R$ centered at a terminal $t$. We have $\sum_{S|v\in\delta(S)} \y(S)=R-(d^c(t,v)-c(v))$.
\end{lemma}


\subsection{An algorithm for the PCSF problem}\label{sec:pcsf}
The algorithm finds the solution $X$ iteratively.
Let $X_i$ denote the set of vertices bought after iteration $i$ where $X_0$ is the set of terminals. For every $i$, the \emph{modified cost function} $c_i$ is a copy of $c$ induced by setting the cost of vertices in $X_{i-1}$ to zero, i.e., $c_i=c[X_{i-1}\to 0]$.
At iteration $i$ there is a set of \emph{active} demands $\Mc{L}_i\subseteq \Mc{L}$ and the dual program $\Mc{D}_i=\overline{\Mc{D}(c_i,\Mc{L}_i)}$. The program $\Mc{D}_i$ is the simplified dual program w.r.t. the modified cost function and the active demands. Note that $\Mc{D}_i$ is stricter than \ref{lp:dual}, thus the objective value of a feasible solution to $\Mc{D}_i$ is a lower bound for $\OPT$. The algorithm guarantees that for every $i<j$, $X_i\subseteq X_j$ and $\Mc{L}_i$ is a superset of $\Mc{L}_j$.


The algorithm is as follows (see Algorithm~\ref{alg:gw}). We initialize $X_0=\Tau$, $c_1=c$, and $\Mc{L}_1=\Mc{L}$.
At iteration $i$, consider the cores formed w.r.t. $c_i$ and $\Mc{L}_i$. Let $T_i$ denote a set which has exactly one terminal in each core (so the number of cores is $|T_i|$).
The algorithm finds the maximum radius $R_i$ such that the $|T_i|$ disks of radius $R_i$ centered at each terminal in $T_i$ are non-overlapping w.r.t. $\Mc{D}_i$. By Lemma~\ref{lem:maxR} either the constraint~\ref{eq:prizeconst} is tight for a continent $S$; or the constraint~\ref{eq:dualconst} is tight for a vertex $v$ on the boundary of multiple disks. In the former, deactivate every demand with exactly one endpoint in $\core(S)$; pay the penalty of such demands and continue to the next iteration with the remaining active demands. In the latter, let $L_v$ denote the centers of the disks whose boundaries contain $v$. For every terminal $\tau\in L_v$ buy the shortest path w.r.t. $c_i$ connecting $v$ to $\tau$ (and so to the core containing $\tau$). Deactivate a demand if its endpoints are now connected in the solution and continue to the next iteration.
The algorithm stops when there is no active demand remaining; in which case it returns the final set of vertices bought by the algorithm.

\begin{algorithm}[!ht]
\textbf{Input:} A graph $G=(V,E)$, a set of demands $\Mc{L}$ with penalties, and a cost function $c$.

\begin{algorithmic}[1]
    \STATE Initialize $X_0=\Tau$, $\Mc{L}_1=\Mc{L}$, $c_1=c$, and $i=1$.
    \WHILE{$|\Mc{L}_i|>0$}
        \STATE Set $c_i=c[X_{i-1}\to 0]$ and construct the dual program $\Mc{D}_i$ with respect to $c_i$ and $\Mc{L}_i$.
        \STATE Construct $T_i$ by choosing an arbitrary terminal from each core.
        \STATE Let $R_i$ be the maximum radius such that putting a disk of radius $R_i$ centered at every terminal in $T_i$ is feasible w.r.t. $\Mc{D}_i$.
        \IF{the constraint~\ref{eq:prizeconst} is tight for a continent $S$}
            \STATE Set $X_i=X_{i-1}$.
            \STATE Set
            $\Mc{L}_{i+1}=\Mc{L}_i\backslash\left\{ j\in[h] | \textit{ either } s_j\in\core(S)\textit{ or } t_j\in\core(S) \right\}$.
            \label{com:gw:prize}
        \ELSE
            \STATE Find a vertex $v$ on the boundary of multiple disks for which constraint~\ref{eq:dualconst} is tight.
            \STATE Let $L_v$ denote the centers of the disks whose boundaries contain $v$.
            \STATE Initialize $X_i=X_{i-1}$.
            \FORALL{$\tau\in L_v$}
                \STATE Add the shortest path (w.r.t. $c_i$) between $\tau$ and $v$ to $X_i$. \label{com:gw:cost}
            \ENDFOR
            \STATE Set $\Mc{L}_{i+1}=\Mc{L}_i\backslash\left\{ j\in[h] | d^{c_{i+1}}(s_j,t_j)=0 \right\}$. \label{com:gw:connected}
        \ENDIF
        \STATE $i=i+1$.
    \ENDWHILE
    \STATE Output $X_{i-1}$.
\end{algorithmic}
\caption{The Prize-Collecting Steiner Forest Algorithm}
\label{alg:gw}
\end{algorithm}

We bound the objective cost of the algorithm in each iteration separately.
The following theorem shows that the fraction of $\OPT$ we incur at each iteration is proportional to the reduction in the number of cores after the iteration.

\begin{theorem}\label{thm:gwalg}
The approximation ratio of Algorithm~\ref{alg:gw} is at most $2 H_{2h}$ where $H_{2h}$ is the $(2h)^{th}$ harmonic number.
\end{theorem}
\begin{proof}
Observe that at each iteration, a core is a connected component of the solution which contains an endpoint of at least one active demand. We distinguish between two types of iterations: In Type I, Line~\ref{com:gw:prize} of Algorithm~\ref{alg:gw} is executed while in Type II, Line~\ref{com:gw:connected} is executed.

Observe that a demand is deactivated either at Line~\ref{com:gw:prize} or at Line~\ref{com:gw:connected}. In the latter, the endpoints of a demand are indeed connected in the solution. Thus we only need to pay the penalty of a demand if it is deactivated in an iteration of Type I. Recall that at Line~\ref{com:gw:prize}, the penalty of $\core(S)$ is half the total penalty of demands cut by $S$. Thus the total penalty we incur at that line is exactly $2 \pi_{\Mc{L}_i}(S)$


We now break the total objective cost of the algorithm into a payment $P_i$ for each iteration $i$ as follows:
\begin{displaymath}
P_i = \left\{ \begin{array}{ll}
2\pi_{\Mc{L}_i}(S) & \textrm{for Type I iterations executing Line~\ref{com:gw:prize} with the continent $S$;}\\
c(X_i)-c(X_{i-1}) & \textrm{for Type II iterations.}
\end{array} \right.
\end{displaymath}
Recall that $|T_i|$ is the number of cores at iteration $i$.
Observe that by Fact~\ref{fact:diskvalue}, at iteration $i$ the total dual vector has value $R_i|T_i|$. By the weak duality $R_i\leq \frac{\OPT}{|T_i|}$.
For every $i\geq 1$, let $h_i=|T_i|-|T_{i+1}|$ denote the reduction in the number of cores after the iteration $i$.

\begin{claim}
$P_i\leq 2 h_i R_i$ for every iteration $i$.
\end{claim}
\begin{proof}
Fix an iteration $i$. Let $\y$ denote the union of disks of radius $R_i$ centered at $T_i$. We distinguish between the two types of the iteration:
\begin{itemize}
\item \emph{Type I}. At Line~\ref{com:gw:prize}, by deactivating all the demands crossing a core, we essentially remove that core. Thus in such an iteration $h_i=1$. The objective cost of the iteration is $2\pi_{\Mc{L}_i}(S)$. On the other hand, the constraint~\ref{eq:dualconst} is tight for $S$, i.e., $\sum_{S'\subseteq S} \y(S)=\pi_{\Mc{L}_i}(S)$. By Facts~\ref{fact:diskvalue} and \ref{fact:outside}, the radius $R_i$ equals $\sum_{S'\subseteq S} \y(S)$. Therefore the objective cost is at most $2 h_i R_i$
\item \emph{Type II}. At line~\ref{com:gw:connected}, we connect $|L_v|$ cores to each other, thus reducing the number of cores in the next iteration by at least $h_i\geq |L_v|-1$%
\forAbs{.}
\forFull{\footnote{In the special case that every endpoint in the cores become connected to the other endpoint of its demand, $h_i=|L_v|$; otherwise $h_i=|L_v|-1$.}.}
    Recall that by Lemma~\ref{lem:maxR}, $|L_v|\geq 2$ and hence $h_i\geq 1$.
    The total cost of connecting terminals in $L_v$ to $v$ is bounded by $c_i(v)$ plus for every $\tau\in L_v$, the cost of the path connecting $\tau$ to $v$ excluding $c_i(v)$. Thus $P_i\leq c_i(v)+\sum_{\tau\in L_v} (d^{c_i}(\tau,v)-c_i(v))$. Now we write the equation in Lemma~\ref{lem:boundarydistance} for every disk centered at a terminal in $L_v$:
    \begin{align*}
    |L_v| R_i  &= \sum_{\tau\in L_v} [d^{c_i}(\tau,v)-c_i(v) + \sum_{S| v\in\delta(S), \tau\in S} \y(S)  ]\\
        &= \sum_{\tau\in L_v} \left[d^{c_i}(\tau,v)-c_i(v) \right] + c_i(v) \geq P_i,
    \end{align*}
    where the last equality follows since the constraint~\ref{eq:dualconst} is tight for $v$. Since the disks are non-overlapping, by Fact~\ref{fact:outside}, $\y(S)$ is positive only if it contains a single terminal of $L_v$. This completes the proof since $P_i\leq |L_v| R_i \leq (h_i+1)R_i \leq 2 h_i R_i$. \qedhere
\end{itemize}
\end{proof}

Let $X$ be the final solution of the algorithm. Note that $|T_{i+1}|=|T_i|-h_i$ and $|T_1|\leq |\Tau|$. A simple calculation shows
\vspace{-0.25cm}
\begin{align*}
\pcsf(X)\leq \sum_i P_i &\leq \sum_i 2h_i R_i \leq 2\OPT \sum_i \frac{h_i}{|T_i|} \leq 2\OPT\cdot H_{|\Tau|}. \qedhere
\end{align*}
\end{proof}


\section{The Budgeted Steiner Tree Problem} \label{sec:budgeted}
In this section we consider the Budgeted problem in the node-weighted Steiner tree setting. Recall that for a vertex $v\in V$, we denote the prize and the cost of the vertex by $\pi(v)$ and $c(v)$, respectively. First we generalize the trimming process of Guha et al.~\cite{Guha99} which reduces the budget violation of a solution while preserving the prize-to-cost ratio. We use this process to obtain a bicriteria approximation algorithm for the rooted version in Section~\ref{sec:rootedbudgeted}.
Next, in Section~\ref{sec:unrootedbudgeted} we consider the unrooted version. By providing a structural property of near-optimal solutions, we propose an algorithm which achieves a logarithmic approximation factor without violating the budget constraint; improving on the previous result of Guha et al.~\cite{Guha99} which obtains an $O(\log^2 n)$-approximation algorithm without violation.

In what follows, for a rooted tree $T$ we assume a \emph{subtree} rooted at a vertex $v$ consists of all vertices whose path to the root of $T$ passes through $v$. The set of \emph{strict subtrees} of $T$ consists of all subtrees other than $T$ itself. Further, the set of \emph{immediate subtrees} of $T$ are the subtrees rooted at the children of the root of $T$.

\subsection{The Rooted Budgeted Problem}\label{sec:rootedbudgeted}
For a budget value $B$ and a vertex $r$, a graph is \emph{$B$-proper for the vertex $r$} if the cost of reaching any vertex from $r$ is at most $B$. The following lemma shows a bicriteria trimming method\forAbs{ (proof in the full version)}.

\begin{lemma}\label{lem:trimmingeps}
Let $T$ be a subtree rooted at $r$ with the prize-to-cost ratio $\gamma$. Suppose the underlying graph is $B$-proper for $r$ and for $\epsilon\in(0,1]$ the cost of the tree is at least $\frac{\epsilon B}{2}$. One can find a tree $T^*$ containing $r$ with the prize-to-cost ratio at least $\frac{\epsilon}{4}\gamma$ such that $\frac{\epsilon}{2} B\leq c(T^*) \leq (1+\epsilon) B$.
\end{lemma}
\forFull{
\begin{proof}
Consider $T$ rooted at $r$. As an initial step, we repeatedly remove a subtree of $T$ if (i) the (prize-to-cost) ratio of the remaining tree is at least $\gamma$; and (ii) the cost of the remaining tree is at least $\frac{\epsilon B}{2}$. We repeat this until no such subtree can be found.

If the current cost of $T$ is at most $(1+\epsilon)B$ we are done. Suppose it is not the case.
A subtree $T'$ is \emph{rich} if $c(T')\geq \frac{\epsilon}{2} B$ and the ratio of $T'$ and all its subtrees is at least $\gamma$. Indeed the existence of a rich subtree proves the lemma.
\begin{claim}\label{claim:richsubtree}
Given a rich subtree $T'$, the desired tree $T^*$ can be found.
\end{claim}
\begin{proof}
Find a rich subtree $T''\subseteq T'$ such that the strict subtrees of $T''$ are not rich, i.e., $c(T'')\geq \frac{\epsilon}{2} B$ while the cost of strict subtrees of $T''$ (if any exists) is less than $\frac{\epsilon}{2}B$. Let $C$ denote the total cost of the immediate subtrees of $T''$.
We distinguish between two cases.
\begin{itemize}
\item If $C<\frac{\epsilon}{2} B$ then we can connect the root of $T''$ directly to $r$. The cost of the resulting tree is at most $C+B\leq (1+\epsilon) B$. On the other hand, $T''$ is rich thus the prize of $T''$ is at least $\gamma\left(\frac{\epsilon}{2} B\right)$. Therefore the resulting tree has the desired ratio $\frac{\gamma \epsilon}{2(1+\epsilon)}\geq \frac{\gamma \epsilon}{4}$.
\item If $C\geq \frac{\epsilon}{2} B$, we can pick a subset of immediate subtrees of $T''$ such that their total cost is between $\frac{\epsilon}{2}B$ and $\epsilon B$. We connect these subtrees to the root by picking the path from the root of $T''$ to $r$. Using the same argument as above, one can show that the resulting tree has the desired properties. \qedhere
\end{itemize}
\end{proof}

It only remains to consider the case that no rich subtree exists. Since $T$ is not rich, the ratio of at least one subtree is less than $\gamma$. Find a subtree $T'$ such that the ratio of $T'$ is less than $\gamma$ while the ratio of all of its strict subtrees (if any exists) is at least $\gamma$. Though the ratio of $T'$ is low, we have not removed it in the initial step. Thus the cost of $T\backslash T'$ is less than $\frac{\epsilon}{2}B$. However, $c(T)>(1+\epsilon) B$ and thus the total cost of immediate subtrees of $T'$ is at least $\frac{\epsilon}{2}B$. On the other hand the cost of an immediate subtree of $T'$ is less than $\frac{\epsilon}{2}B$, otherwise it would be a rich subtree. Therefore we can pick a subset of immediate subtrees of $T'$ such that their total cost is between $\frac{\epsilon}{2}B$ and $\epsilon B$. We connect these subtrees by connecting the root of $T'$ directly to $r$. The resulting tree has the cost at most $(1+\epsilon) B$ and the prize at least $\gamma\left(\frac{\epsilon}{2} B\right)$ which completes the proof.
\end{proof}
}

Moss and Rabani~\cite{MossRabani} give an $O(\log n)$-approximation algorithm for the Budgeted problem which may violate the budget by a factor of two.
Using Lemma~\ref{lem:trimmingeps} one can trim such a solution to achieve a trade-off between the violation of budget and the approximation factor \forAbs{(proof in the full version)}.

\begin{theorem}\label{thm:rootedeps}
For every $\epsilon\in(0,1]$ one can find a subtree $T\subseteq G$ in polynomial time such that $c(T)\leq (1+\epsilon) B$ and the total prize of $T$ is $\Omega(\frac{\epsilon^2}{\log n})$ fraction of $\OPT$.
\end{theorem}
\forFull{
\begin{proof}
First we make the graph $B$ proper for the root $r$ by simply discarding the vertices which are farther than $B$ from $r$. Note that these vertices cannot be a part of an optimal solution.
By Theorem~\ref{thm:mossbudgeted}, we can find a tree $T'$ with $\pi(T')\geq \frac{\OPT}{O(\log n)}$ and $c(T')\leq 2B$. Suppose the cost of $T'$ is more than $(1+\epsilon)B$; otherwise we are done.

Let $\gamma(T')$ denote the prize to cost ratio of $T'$. Observe that $\gamma(T')=\frac{\pi(T')}{c(T')}\geq \frac{\OPT}{O(\log n) \cdot 2B}$. By Lemma~\ref{lem:trimmingeps} we can trim $T'$ to obtain a subtree $T$ such that
\begin{itemize}
\item The prize to cost ratio of $T$ is $\gamma(T)\geq \frac{\epsilon}{4}\gamma(T')\geq \frac{\epsilon \OPT}{O(\log n) B}$.
\item The cost of $T$ is sandwiched between $\frac{\epsilon}{2}B$ and $(1+\epsilon)B$.
\end{itemize}
Therefore the cost of $T$ does not violate the budget by much and $\pi(T)$ is at least
\begin{center}
$\pi(T)\geq \gamma(T)\left(\frac{\epsilon}{2}B\right) \geq \frac{\epsilon^2 \OPT}{O(\log n)}$. \qedhere
\end{center}
\end{proof}
} 

\subsection{The Unrooted Budgeted Problem}\label{sec:unrootedbudgeted}
We prove a stronger variant of Lemma~\ref{lem:trimmingeps} for the unrooted version. We show that if no single vertex is too expensive, one does not need to violate the budget at all. The analysis is similar to that of Lemma~\ref{lem:trimmingeps}. \forFull{For the sake of completeness, we have presented the proof in Appendix~\ref{sec:omitted}.}

\begin{lemma}\label{lem:trimunrooted}
Let $T$ be a tree with the prize to cost ratio $\gamma$. Suppose $c(T)\geq \frac{B}{2}$ and the cost of every vertex of the tree is at most $\frac{B}{2}$ for a real number $B$. One can find a subtree $T^*\subseteq T$ with the prize to cost ratio at least $\frac{\gamma}{4}$ such that $\frac{B}{4}\leq c(T^*) \leq B$.
\end{lemma}

One may use arguments similar to that of Theorem~\ref{thm:rootedeps} to derive an $O(\log n)$-approximation algorithm from Lemma~\ref{lem:trimunrooted} when the cost of a vertex is not too big. On the other hand if the cost of a vertex is more than half the budget, we can guess that vertex and try to solve the problem with the remaining budget. However, one obstacle is that this process may need to be repeated, i.e., the cost of another vertex may be more than half  the remaining budget. Thus we may need to continue guessing many vertices in which case connecting them in an optimal manner would not be an easy task. The following theorem shows indeed guessing one vertex is sufficient if one is willing to lose an extra factor of two in the approximation guarantee.

\begin{theorem}\label{thm:logunrooted}
The unrooted budgeted problem admits an $O(\log n)$-approximation algorithm which does not violate the budget constraint.
\end{theorem}
\begin{proof}
We define two classes of subtrees: the \emph{flat} trees and the \emph{saddled} trees. A tree is flat if the cost of every vertex of the tree is at most $\frac{B}{2}$. For a tree $T$, let $x$ be the vertex of $T$ with the largest cost. The tree $T$ is saddled if $c(x)>\frac{B}{2}$ and the cost of every other vertex of the tree is at most $\frac{B-c(x)}{2}$. Let $T^*_f$ denote the optimal flat tree, i.e., a flat tree with the maximum prize among all the flat trees with the total cost at most $B$. Similarly, let $T^*_s$ denote the optimal saddled tree.

The proof is described in two parts. First we show the prize of the best solution between $T^*_f$ and $T^*_s$ is indeed in a constant factor of $\OPT$\forAbs{(see the following claim, proof in the full version)}.
Next, we show by restricting the optimum to any of the two classes, an $O(\log(n))$-approximation solution can be found in polynomial time. Therefore this would give us the desired approximation algorithm.

\begin{claim}\label{claim:optstructure}
Either $\pi(T^*_f)\geq \frac{\OPT}{2}$ or $\pi(T^*_s)\geq \frac{\OPT}{2}$.
\end{claim}
\forFull{
\begin{proof}
Let $T^*$ denote the optimal tree. If $T^*$ consists of only one vertex then clearly it is either a flat tree or a saddled tree and we are done. Now assume that $T^*$ is neither flat nor saddled and it has at least two vertices. Let $x$ and $y$ denote the vertices with the maximum cost and the second maximum cost in $T^*$, respectively. Since $T^*$ is not flat we have $c(x)>\frac{B}{2}$ and $c(y)\leq \frac{B}{2}$. It is neither saddled thus $c(y)> \frac{B-c(x)}{2}$. Observe that the cost of any other vertex of $T^*$ is at most $\frac{B-c(x)}{2}$. Consider the path between $y$ and $x$ in $T^*$. Let $e$ denote the edge of the path which is adjacent to $y$. Removing $e$ from $T^*$ results in the two subtrees $T_y$ and $T_x$ containing $y$ and $x$, respectively. The cost of every vertex in $T_y$ is at most $c(y)\leq \frac{B}{2}$, thus $T_y$ is flat. On the other hand the cost every vertex in $T_x$ except $x$ is at most $\frac{B-c(x)}{2}$, thus $T_x$ is saddled.
This completes the proof since one of the subtrees has at least half  the optimal prize $\pi(T^*)$.
\end{proof}
}

Now we only need to restrict the algorithm to flat trees and saddled trees. Indeed we can reduce the case of saddled trees to flat trees. We simply guess the maximum-cost vertex $x$ (by iterating over all vertices). We form a new instance of the problem by reducing the budget to $B-c(x)$ and the cost of $x$ to zero. The cost of every other vertex in $T^*_s$ is at most half  the remaining budget, thus we need to look for the best flat tree in the new instance. Therefore it only remains to find an approximation solution when restricted to flat trees.

We use Lemma~\ref{lem:trimunrooted} to find the desired solution for flat trees. A vertex with cost more than half  the budget cannot be in a flat tree, thus we remove all such vertices. We may guess a vertex of the best solution and by using the algorithm of Moss and Rabani~\cite{MossRabani}
\forFull{(see Theorem~\ref{thm:mossbudgeted} in the Appendix)}
we can find an $O(\log n)$-approximation solution which may use twice the budget. Let $T$ be the resulting tree with the total prize $P$. If $c(T)\leq B$ we are done. Otherwise by Lemma~\ref{lem:trimunrooted} we can trim $T$ to obtain a tree with the cost at most $B$ and the prize at least $\frac{P}{32}$ which completes the proof.
\end{proof}


\bibliographystyle{plain}

\forFull{
\appendix
\section{Hardness of Net Worth}\label{sec:nw-hardness}
Here we present the hardness result for the rooted NW and directed NW given in Feigenbaum et al.~\cite{papa01} with slight modifications. We show that NW is NP-hard to approximate within any finite factor when restricted to the case of bounded degree graphs.

\begin{theorem}
For any $\epsilon$, $0<\epsilon<1$, it is NP-hard to approximate\footnote{Algorithm $A$ approximates function $f$ within ratio $\epsilon$, iff for every input instance $x$, $\epsilon f(x)\leq A(x)\leq f(x)/\epsilon$. Since NW is a maximization problem, we can assume that $A(x)\leq f(x)$.}
the rooted, whether directed or undirected, net worth problem within a ratio $\epsilon$.
\end{theorem}
\begin{proof}
Given an instance $I$ of 3-SAT, we make an instance $J$ of NW problem such that: (i) if $I$ is a yes-instance (i.e., it is satisfiable), then an $\epsilon$-approximation answer to $J$ is strictly greater than $\epsilon$; and (ii) if $I$ is a no-instance the the optimal answer to $J$ is at most $\epsilon$.
Let $n$ and $m$ be the number of variables and clauses in $I$, respectively. Without loss of generality we assume that for every variable $x$ there is a clause $x \vee \bar{x}$ in $I$, thus $m\geq n+1$. We make the instance $J$ with four layers of vertices as follows:
\begin{itemize}
\item In the top layer, we put the root $r$ with prize $\pi(r)=\epsilon$.

\item In the second layer, we put a vertex $r'$ with prize zero, connected to $r$ via an edge of cost $m K-(n+1)-m$ for a fix $K\geq n+1$.

\item The third layer contains $2n$ vertices each for every literal in $I$, all with prize zero and connected to $r'$ via edges of unit cost.

\item The last layer contains $m$ vertices for every clause, all with prize $K$ and connected to the vertices corresponding to the literals it contains via edges of unit cost.
\end{itemize}
In the case of directed NW, we direct all the edges from top to bottom.
We claim that if $I$ is satisfiable then $NW(J)\geq 1+\epsilon$, otherwise $NW(J)\leq \epsilon$. Note that in the former an $\epsilon$-approximation algorithm would give us a solution with net worth at least $\epsilon(1+\epsilon)>\epsilon$  and in the latter it would give us a solution with net worth at most $\epsilon$, thus it can distinguish the satisfiability of $I$.

To get a solution with net worth more than $\epsilon$, we have to buy the edge $(r,r')$ and thus incurring a big cost. In order to include a subset $S$ of vertices of the fourth layer, we need to buy at least one edge between layer two and layer three, and $|S|$ edges connecting layer three to $S$. Thus the maximum net worth we could get is
\begin{eqnarray*}
\epsilon+ |S| K-1-|S|-(m K - (n+1) -m) &=& \epsilon+ |S|(K-1)-m(K-1)+n \\
        &=& \epsilon+(|S|-m)(K-1)+n\\
        &\leq& \epsilon+ n(|S|-m+1)
\end{eqnarray*}
where the last inequality holds since $|S|-m\leq 0$ and $K\geq n+1$.
Therefore to have a net worth strictly more than $\epsilon$, we need to include all the vertices in the fourth layer. Observer that the maximum possible net worth is $\epsilon+n$.
Recall that for every variable $x$ there is a clause $x \vee \bar{x}$. To include the vertex corresponding to this clause, we need to include at least one vertex corresponding to a literal of $x$. On the other hand, we cannot include more than $n$ vertices of the third layer or otherwise we could not achieve a net worth more than $\epsilon$. This shows that for every variable $x$, the vertex $r'$ would be connected to exactly one of the vertices corresponding to $x$ and $\bar{x}$. Therefore a solution of net worth more than $\epsilon$ corresponds to a satisfying assignment, in fact, the net worth of such a solution would be exactly $1+\epsilon$. This completes the proof.
\end{proof}

\section{Reductions for Quota Problems}\label{sec:reductions}
Here we present two important reductions that were missing in the literature.
More specifically, we show that rooted and unrooted $k$-MST and their $k$-Steiner Tree versions
are all equivalent and indeed equivalent to the Quota problem.
(That these are simpler than the latter is easy.)
These results improve the approximation ratio of $k$-Steiner Tree from $4$ to $2$.
Ravi et al.~\cite{RaviEtAl} had provided a reduction from $k$-Steiner Tree to $k$-MST
losing a factor $2$, whereas Chudak et al.~\cite{CRW} had conjectured the presence
of a $(2+\eps)$-approximation
algorithm while presenting one with approximation ratio of $5$.



This appendix deals with four special cases of the quota node-weighted Steiner tree problem.
We first claim that the rooted $k$-Steiner tree problem is equivalent to the quota problem,
with a factor of $1+\eps$ for a polynomially small $\eps$.
That the former is a special case of the latter can be observed easily
by setting vertex prizes to $0$ and $1$ for Steiner and terminal nodes, respectively,
and looking for a prize of at least $k$.
To establish the other direction of the reduction,
given a graph $G$ with prize $\pi$ and cost $c$ on its vertices,
as well as target prize value $P$,
we produce an instance of $k$-Steiner tree as follows.
We assume all vertices of $G$ are Steiner vertices and connect
a vertex $u$ to $q(u)=\lceil\frac{n \pi(u)}{\eps P}\rceil$ new terminal vertices of cost zero.
In this instance we let $k = \lfloor n/\eps\rfloor$.
Clearly any solution to the quota instance turns into a solution of $k$-Steiner tree
if one collects the terminals immediately connected to the solution vertices.
Next consider a solution to the $k$-Steiner tree instance.
We can assume without loss of generality that either none or all the terminals
connected to one node are in the solution.
The solution to the quota instance simply includes all Steiner nodes whose all adjacent terminals are picked in the k-Steiner tree instance. For such nodes we have $\sum_u q(u) \geq k$. Note that there are at most $n$ such Steiner nodes, and for each of them, say $u$, we have $q(u)\cdot\frac{\epsilon P}{n} < \frac{\epsilon P}{n} + \pi(u)$. Therefore, we get $(\frac{\epsilon P}{n}) \sum_u q(u) < \epsilon P + \sum_u \pi(u)$. However, the solution guarantee (in the k-MST instance) is that the left-hand side is at least $(\frac{\epsilon P}{n})k > (\frac{\epsilon P}{n}) (\frac{n}{\epsilon} - 1) = P - \frac{\epsilon P }{n}$. Putting these two together and noting that $n \geq 1$, we obtain $\sum_u \pi(u) > P(1 - 2\epsilon)$.

The following two theorems show the other three problems
are equivalent to $k$-Steiner tree (and hence quota problem).

\begin{theorem}\label{thm:MST_rooted_reduction}
Let $\alpha$, $0<\alpha<1$, be a constant. The following two statements are equivalent:
\begin{enumerate}[i]
\item There is an $\alpha$-approximation algorithm for the rooted $k$-\MST{} problem.
\item There is an $\alpha$-approximation algorithm for the unrooted $k$-\MST{} problem.
\end{enumerate}
\end{theorem}
\begin{proof}
We note that by running the rooted $k$-\MST{} for every vertex, (i) immediately implies (ii). To prove that (ii) implies (i), we give a cost-preserving reduction from rooted variant to unrooted variant. Let $<G=(V,E),r,k>$ be an instance of the rooted $k$-\MST{} and let $n=|V|$. We add $n$ vertices to $G$, all connected by edges of cost zero to $r$. Let $k'=k+n$ and consider the solution to (unrooted) $k'$-\MST{} on the new graph. Since $k'>n-1$, a subtree of size $k'$ has to include $r$. Thus we can assume that there exist an optimal solution which includes all the $n$ extra vertices plus a minimum-cost subtree of size $k$ rooted at $r$. Hence the reduction preserves the cost of optimal solution.
\end{proof}

\begin{theorem}\label{thm:MST-ST}
Let $\alpha$, $0<\alpha<1$, be a constant. The following two statements are equivalent:
\begin{enumerate}[i]
\item There is an $\alpha$-approximation algorithm for the $k$-\ST{} problem.
\item There is an $\alpha$-approximation algorithm for the $k$-\MST{} problem.
\end{enumerate}
\end{theorem}
\begin{proof}
We note that one way is clear by definition. To prove that (ii) implies (i), similar to Theorem~\ref{thm:MST_rooted_reduction}, we give a cost-preserving reduction from $k$-\ST{} to $k$-\MST{}. Let $<G=(V,E),T,k>$ be an instance of $k$-\ST{} with the set of terminals $T\subseteq V$. Let $n=|V|$. For every terminal $v_t\in T$, add $n$ vertices at distance zero of $v_t$. Let $k'=kn+k$ and consider the solution to $k'$-\MST{} on the new graph. Any subtree with at most $k-1$ terminals have at most $(k-1)n+n-1=kn-1$ vertices. Therefore an optimal solution covers at least $k$ terminals. Hence the reduction preserves the cost of optimal solution.
\end{proof}

All the above proofs work, mutatis mutandis, for the edge-weighted case, too.

\section{Integrality Gap for Budgeted Steiner Tree}\label{sec:natural-lp}
In this section we discuss the linear programming approach to the Budgeted problem. Let $P_v$ denote the set of all paths from root to vertex $v$. We may also assume that all the edges have unit length. Consider the flow-based linear programming below.

\noindent
\begin{align*}
    & \text{maximize.}                  & &\sum_{v\in V} \pi_v \sum_{p\in P_v} \f_p \\
    &  \forall e\in E, v\in V&     & \sum_{p\in P_v : e\in p} \f_p \leq \x_e \tag{X}\label{cons:X}\\
    &  \forall v\in V&             & \sum_{p\in P_v} \f_p \leq 1 \tag{F}\label{cons:F}\\
    &   &   &   \sum_{e\in E} \x_e \leq B \tag{B}\label{cons:B}\\
    &   &   &   \f_p, \x_e \geq 0
\end{align*}
Intuitively, for a path $p$ ending at $v$, $\f_p$ denote the total flow reaching $v$ through $p$ and $\x_e$ denote the maximum flow passing through the edge $e$. Constraint \ref{cons:B} keeps the cost of edges in budget and constraint \ref{cons:F} restricts the total flow reaching a vertex. One can also write a similar cut-based linear programming. However, we can show that even if $G$ is a tree, the gap between the fractional and integral solutions is unbounded. Let $G$ be a tree obtained by putting a star at the end of a long path of length $B-1$ (see Fig.~\ref{fig:gap}). Let $u_1,\ldots, u_k$ denote the leaves other than the root which have $1$ unit of profit. Other vertices have zero profit. Clearly the optimal integral solution gains one unit of profit. Let $p_i$ denote the path from $r$ to $u_i$. Consider a feasible fractional solution where for every $i$, $f_{p_i}=\frac{B}{B+k-1}$ and therefore for every edge $e$, $x_e=\frac{B}{B+k-1}$. We note that since there are $B+k-1$ edges, we are not exceeding the budget. This shows that the optimal fractional solution is at least $\frac{kB}{B+k-1}$ and hence in case of $B\geq k$, the gap between the fractional and the integral solution is $k$.

\begin{figure}
    \begin{center}
        \includegraphics[width=0.9\textwidth]{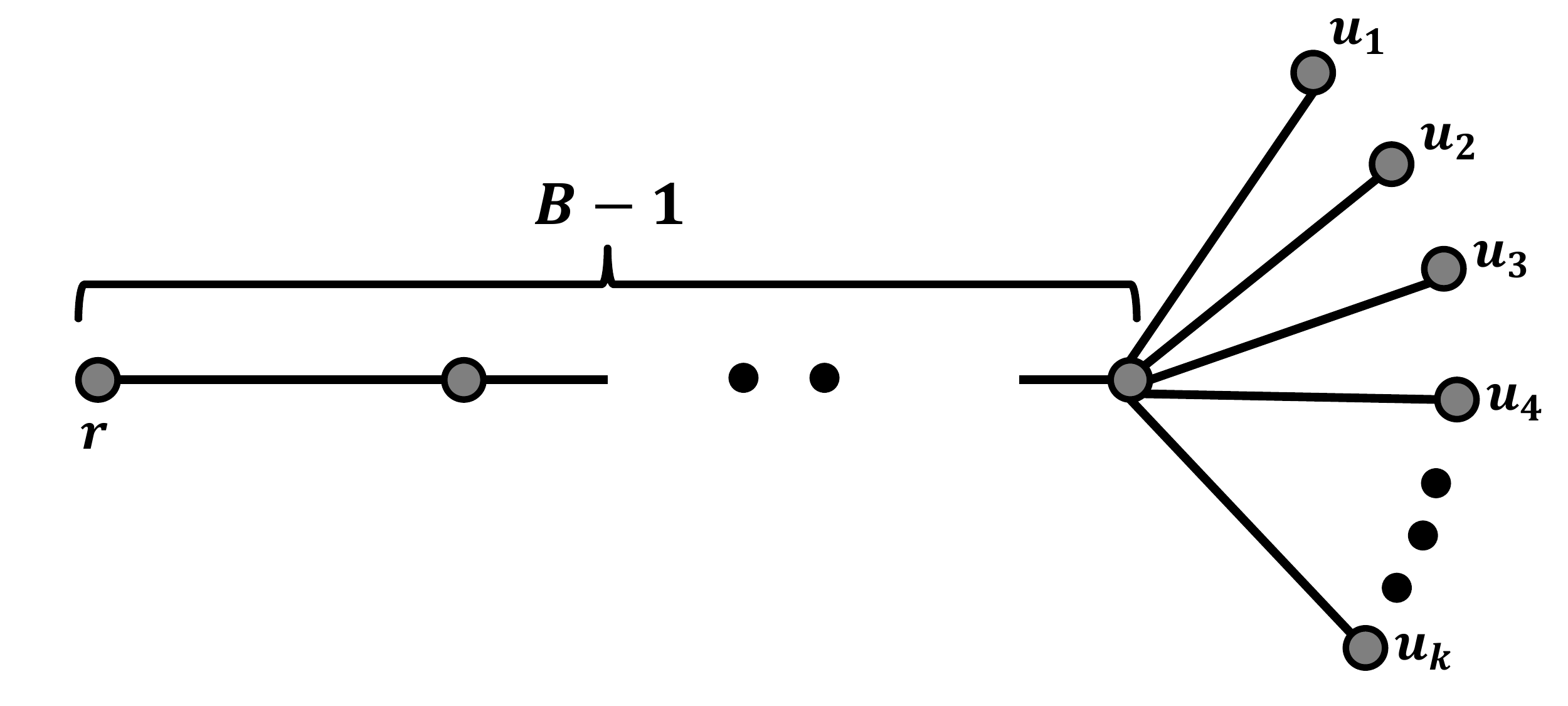}
    \end{center}
  \caption{An example showing the unbounded gap of the LP for the budget problem.}
        \label{fig:gap}
\end{figure}

\section{Omitted Proofs}\label{sec:omitted}
\begin{proof}[Proof of Lemma~\ref{lem:maxR}]
Let $\y_1,\ldots,\y_{|T'|}$ denote the disks of radius $R$ centered at the terminals in $T'$. Increasing the radius of all disks by any $\epsilon>0$ creates an infeasibility in their union $\y$. Thus at least one of the followings holds for $\y$:
\begin{itemize}
\item \emph{The constraint~\ref{eq:prizeconst} is tight for a set $S$ containing a terminal, i.e., $\sum_{S'\subseteq S} \y(S')=\pi(S)>0$.} Let $S$ be such a set with the smallest cardinality. Recall that by Fact~\ref{fact:outside}, $\y(S')$ is positive for a set $S'$ only if $S'$ contains the center of a disk and is a subset of the continent of that disk. We remove the zero terms from both sides of the equality. The right hand side would be the penalty of a subset of the terminals and the left hand side would be the sum over dual variables $\y(S')$'s such that $S'$ is a subset of a continent. If the inequality is tight it has to be tight induced to any disk, otherwise $\y$ is not feasible. Thus the smallest set $S$ is indeed a subset of the continent of a disk. Now let $S^*$ be the continent of that disk. The sets $S$ and $S^*$ share the same core, thus the right hand sides of the constraint~\ref{eq:prizeconst} for both are the same.
    However the the left hand side of the constraint for $S^*$ can only be larger which leads to (i).
\item \emph{For a vertex $v$ the constraint~\ref{eq:dualconst} becomes infeasible if we grow every disk by any $\epsilon>0$.} The constraint for $v$ is tight w.r.t. $\y$. If the constraint for $v$ is not tight in any of $\y_i$'s independently, then $v$ is on the boundary of more than one disk which leads to (ii). Otherwise assume the constraint for $v$ is tight in $\y_i$ for an $i\in[|T'|]$\footnote{For an integer $x$, let $[x]$ denote the set $\{1,2,\ldots, x\}$.}. If we extend the radius by $\epsilon$, $v$ will be inside the $i^{th}$ disk thus $\sum_{S|v\in \delta(S)} \y_1(S)$ will not change. However by the assumption about $v$, the same summation for $\y$, i.e., $\sum_{S|v\in \delta(S)} \y(S)$ will increase. Therefore a neighbor of $v$ is at most $R$ far from the center of another disk, say that of the $j^{th}$ disk. By definition, $v$ is on the boundary of the $j^{th}$ disk. Further, by Proposition~\ref{prop:inorout}, $v$ cannot be inside the $i^{th}$ disk and so is on its boundary which leads to (ii). \qedhere
\end{itemize}
\end{proof}

\begin{proof}[Proof of Lemma~\ref{lem:boundarydistance}]
Consider the process of growing the disk during which we start with a set $S$ (initialized to the core containing $t$) and add the vertices for which the constraint~\ref{eq:dualconst} becomes tight. Using induction it is easy to show that a vertex $u$ in the continent is added to $S$ when the total growth passes $d^c(t,u)$.
Let $u$ be the closest neighbor of $v$ to the center, thus $R-(d^c(t,v)-c(v))=R-d^c(t,u)$. If $u$ is not inside the disk, then $d^c(t,u)=R$ and we are done. Otherwise, as soon as the radius of the disk passes $d^c(t,u)$, $S$ \emph{touches} the vertex $v$ and any further growth contributes to $\sum_{S|v\in\delta(S)} \y(S)$. Since at the end $v$ is not inside the disk, this contribution continues until the total growth reaches $R$. Therefore
\begin{align*}
\sum_{S|v\in\delta(S)} \y(S)=R-d^c(t,u)=R-(d^c(t,v)-c(v)) &\qedhere
\end{align*}
\end{proof}

\begin{theorem}[Theorem 12 of \cite{MossRabani}]\label{thm:mossbudgeted}
For an instance of the rooted budgeted problem, an $O(\log n)$-approximation solution can be found in polynomial time which uses at most twice the budget.
\end{theorem}

\begin{proof}[Proof of Lemma~\ref{lem:trimunrooted}]
We make $T$ rooted at an arbitrary vertex $r$.
As the first pruning step, we repeatedly discard a subtree if the ratio and the cost of the remaining tree does not go below $\gamma$ and $\frac{B}{4}$, respectively. We stop when no such subtree can be found. Suppose the current cost of $T$ is more than $B$; otherwise we are done. As in Lemma~\ref{lem:trimmingeps}, a subtree $T'$ is \emph{rich} if the ratio of $T'$ and all subtrees of $T'$ is at least $\gamma$. Note that one can easily check whether a subtree is rich.

First we show given a rich subtree we can easily find the solution. Observe that all the subtrees of a rich subtree are also rich unless their cost is less than $\frac{B}{4}$. Given a rich subtree, let $T'$ be its lowest rich subtree, i.e., the cost of any immediate subtree of $T'$ (if any exist) is less than $\frac{B}{4}$. Now let $C$ denote the total cost of all immediate subtrees of $T'$.
\begin{itemize}
\item If $C<\frac{B}{4}$ (or no child exists), then $c(T')\leq \frac{3B}{4}$ since the cost of the root of $T'$ does not exceed $\frac{B}{2}$. Thus $T'$ satisfies the properties desired in the lemma. Recall that $T'$ is rich and thus its ratio is at least $\gamma$.
\item If $C\geq\frac{B}{4}$, we can pick a subset of immediate subtrees of $T'$ such that their total cost is between $\frac{B}{4}$ and $\frac{B}{2}$. This can be done since the cost of an immediate subtree is at most $\frac{B}{4}$. Let $T^*$ be the tree formed by connecting these subtrees to the root of $T'$. Observe that $c(T^*)\leq B$ and the total prize is at least $\pi(T^*)\geq \gamma \frac{B}{4}$. Therefore the ratio of $T^*$ is at least $\frac{\gamma}{4}$.
\end{itemize}
It only remains to consider the case that $T$ does not have a rich subtree. Since $T$ is not rich, a subtree of $T$ has ratio less than $\gamma$.
Let $T'$ be a subtree with ratio less than $\gamma$ such that all strict subtrees of $T'$ (if any exists) have ratio at least $\gamma$. Observe that the cost of an immediate subtree of $T'$ is less than $\frac{B}{4}$, otherwise it would be a rich subtree. On the other hand, we have not discarded $T'$ in the first pruning step, hence $c(T\backslash T')<\frac{B}{4}$. Furthermore $c(T)>B$, thus the total cost of immediate subtrees of $T'$ is at least $\frac{B}{4}$.
Now similar to the previous argument, we can pick a subset of immediate subtrees of $T'$ such that their total cost is between $\frac{B}{4}$ and $\frac{B}{2}$. The tree formed by connecting these subtrees to the root of $T'$ has the desired properties.
\end{proof}

}

\end{document}